\documentclass[conference]{IEEEtran}

\usepackage{ieeefig}
\usepackage{latexsym}
\usepackage{graphicx}
\usepackage{array}
\usepackage{amsmath}
\usepackage{color}

\newtheorem{thm}{Theorem}
\newtheorem{cor}{Corollary}

\newtheorem{lemma}{Lemma}

\begin{document}

\newcommand{\yellow}[1]{\colorbox{yellow}{#1}}

\title{Optimum Power and Rate Allocation for Coded V-BLAST}

\author{\IEEEauthorblockN{Victoria Kostina}
\IEEEauthorblockA{Department of Electrical Engineering\\ Princeton University\\
Princeton, NJ, 08544, USA\\ Email: vkostina@princeton.edu}
\and
\IEEEauthorblockN{Sergey Loyka}
\IEEEauthorblockA{School of Information Technology and Engineering\\
University of Ottawa\\ Ontario, Canada, K1N 6N5\\ Email: sergey.loyka@ieee.org}}


\thanks{V. Kostina is with the Department of Electrical Engineering, Princeton University, NJ, 08544, USA (e-mail:
vkostina@princeton.edu)}

\thanks{S. Loyka is with the School of Information Technology and Engineering, University of Ottawa, Ontario,
Canada, K1N 6N5, e-mail: sergey.loyka@ieee.org}

\maketitle

\begin{abstract}
An analytical framework for minimizing the outage probability of a coded spatial multiplexing system while keeping the rate close to the capacity is developed. Based on this framework, specific strategies of optimum power and rate allocation for the coded V-BLAST architecture are obtained and its performance is analyzed. A fractional waterfilling algorithm, which is shown to optimize both the capacity and the outage probability of the coded V-BLAST, is proposed. Compact, closed-form expressions for the optimum allocation of the average power are given. The uniform allocation of average power is shown to be near optimum at moderate to high SNR for the coded V-BLAST with the average rate allocation (when per-stream rates are set to match the per-stream capacity). The results reported also apply to multiuser detection and channel equalization relying on successive interference cancelation.
\end{abstract}


\section{Introduction}

\PARstart{T}{o} exploit the impressive spectral efficiencies of wireless communication systems with multiple antennas at both the transmitter
and receiver \cite{FosLimits}, the V-BLAST algorithm was proposed \cite{FosBLAST}. Its simple transmission and detection mechanisms as well as its ability to achieve a significant portion of the MIMO capacity have made the V-BLAST a popular solution for MIMO signal processing. In this paper, we consider zero-forcing (ZF) V-BLAST, which relies on the successive interference cancelation (SIC) to decode the spatially-multiplexed sub-streams at the receiver. Because of SIC, the algorithm suffers from the error propagation effect so that the overall error performance is dominated by that of the 1st stream (with low diversity order), which may not be satisfactory.

Several techniques have been reported to improve the error performance of the uncoded V-BLAST by employing a non-uniform power allocation among the transmitters \cite{Kost}-\cite{Wang}. References \cite{Kalbasi}-\cite{Wang} explore the transmit (Tx) power allocation that minimizes the instantaneous (i.e. for given channel realization) error rate of the uncoded V-BLAST, with or without the optimal ordering. For this optimization, a new feedback session and power reallocation are needed each time the channel changes. A less demanding approach is to use an average optimization, i.e. to find the optimum allocation of the average power based on the average error rate. Since this ignores the small-scale
fading, only occasional feedback sessions and power reallocations are required, when the average SNR changes, and only the average SNR needs to be fed back to the Tx end. This approach has been exploited in \cite{Kost} and \cite{Varanasi}. None of the references above consider coding and none, with the exception of \cite{Varanasi}, consider rate allocation.

Most practical communication systems use coding; uncoded systems are rare. This motivates an analysis and optimization of the coded V-BLAST. When powerful temporal codes are used for each sub-stream of a spatial multiplexing system, the error probability is very low if the system is not in outage, so that the overall error probability is dominated by the outage events \cite{Tse}. Following this philosophy, we assume that capacity-achieving temporal codes are used for each stream, so that the per-stream rate equals to the capacity of that stream and there are no errors when stream is not in outage, and all bits are in errors otherwise, and study optimum power and rate allocation for the coded V-BLAST\footnote{This approach can also be extended to realistic codes by using the SNR gap to capacity.}. This model of coded V-BLAST allows analytically-tractable optimization and performance analysis of the algorithm.

Our approach to rate allocation is completely different to that of \cite{Varanasi}. Prasad and Varanasi \cite{Varanasi} find such rate allocation among the transmitters that minimizes the average error rate for a fixed total rate $R$. However, it is not optimal in terms of the capacity: when the SNR and hence the channel capacity increase, the total transmission rate $R$ stays the same, i.e. well below the channel capacity at high SNR. While such allocation policy maximizes the diversity gain, it also results in only one transmitter be active at high SNR \cite{Varanasi}, which also shows its sub-optimality capacity-wise (recall that the capacity-optimum strategy at high SNR is to use all transmitters).

In our approach to instantaneous optimization, the total rate and per-stream rates are adjusted to match the capacity and, while the outage probability is used as the main performance criterion, the transmission rate also stays close to the capacity at any SNR. We demonstrate that the conventional waterfilling (WF) algorithm does not maximize the instantaneous system capacity of the coded V-BLAST (due to successive interference cancelation) and propose a new algorithm termed "fractional waterfilling" (FWF) that maximizes the capacity by waterfilling on a sub-set of all transmitters and using an optimum rate allocation among the streams. While the complexity of the proposed FWF is higher compared to the conventional WF, the incremental complexity is small when the number of transmitters is not large. The FWF is shown to converge to the conventional WF at high SNR, and the optimum power allocation converges to the uniform one, which is dramatically different from the average power/rate allocations for the uncoded V-BLAST where most of the power goes to a single transmitter \cite{Varanasi}\cite{Kost}. When the rate allocation is uniform, however, the optimum average power allocation is close to that in the uncoded V-BLAST, which is demonstrated by deriving a compact closed-form expression for the coded case.

We show that while the optimum power allocation for the coded V-BLAST with uniform rate allocation brings only a few dB SNR gain (similarly to the uncoded system \cite{Kost}), the optimum rate allocation brings additional diversity order (even when the power allocation is uniform) and, thus, is much more superior at high SNR. The optimum power allocation on top of the rate allocation brings only a fixed SNR gain.

We also establish the relationship between outage probabilities achieved with various optimization strategies for a broad class of systems and channels. In particular, the maximization of the instantaneous system capacity is shown to also minimize the outage probability and, thus, the two problems are equivalent. The importance of the latter result lies in the fact that while the minimization of the outage probability is very challenging and highly non-convex problem with multiple solutions (as we demonstrate), the maximization of the instantaneous capacity is convex and has a well-known solution (via waterfilling). We show that, while there is a number of strategies to minimize the outage probability, only the FWF simultaneously minimizes the outage probability and maximizes the system capacity of the coded V-BLAST.

Due to the similar system architectures and processing strategies, most of these results also apply to multiuser detection and inter-symbols equalization systems that use successive interference cancellation.

The paper is organized as follows. Section \ref{sec:system} introduces the basic system model, assumptions and optimization strategies. Section \ref{sec:framework} presents a comparison of various strategies to minimize the outage probability of a spatial multiplexing system (without any specific channel assumptions). Sections \ref{sec:OPA} - \ref{sec:OPRA} derive and analyse the optimum power, rate and joint power/rate allocations for the coded V-BLAST. Section \ref{sec:conclusion} concludes the paper.

\section{System Model and Optimization Strategies}
\label{sec:system}

We study zero-forcing (ZF) V-BLAST without optimal ordering but with capacity-achieving temporal codes for each stream so that the maximum possible rate equals to the capacity of that stream (this follows the philosophy in \cite{Tse}). There are no errors if the stream is not in outage and all errors when it is, so that there is no error propagation when all streams are not in outage.

The following standard baseband discrete-time MIMO system model is employed,
\begin{equation}
\label{eqSysModel}
{\rm {\bf r}}={\rm {\bf H \bf \Lambda \bf s}}+ \boldsymbol \xi=\sum\nolimits_{i=1}^m
{{\rm {\bf h}}_i \sqrt {\alpha _i } s_i + \boldsymbol \xi}
\end{equation}
where ${\rm {\bf s}}=[s_1 ,s_2 ,...s_m ]^T$ and ${\rm {\bf r}}=[r_1 ,r_2
,...r_m ]^T$ are the vectors representing the Tx and Rx symbols
respectively, ``$T$'' denotes transposition, ${\rm {\bf H}}=[{\rm {\bf h}}_1
,{\rm {\bf h}}_2 ,...{\rm {\bf h}}_m ]$ is the $n\times m$ matrix of the
complex channel gains between each Tx and each Rx antenna, where ${\rm {\bf
h}}_i $ denotes i-th column of ${\rm {\bf H}}$, $n$ and $m$ are the numbers of Rx
and Tx antennas respectively, $n\ge m$, $\boldsymbol \xi$ is the vector of
circularly-symmetric additive white Gaussian noise (AWGN), which is
independent and identically distributed (i.i.d.) in each receiver, $\boldsymbol \Lambda = diag\left( {\sqrt {\alpha _1 } ,\ldots ,\sqrt {\alpha _m } } \right)$,
where $\alpha _i $ is the power allocated to the $i$-th transmitter. For the
regular V-BLAST, the total power is distributed uniformly
among the transmitters, $\alpha _1 =\alpha _2 =...=\alpha _m =1$. The channel will be assumed to be either ergodic ("fast fading"), in which case the key performance measure is the ergodic system capacity, or non-ergodic ("slow fading"), in which case the key performance measures are the outage probability and outage capacity and also an instantaneous system capacity (for given channel realization) \cite{Tse}. Details of a mathematical model of the uncoded V-BLAST, on which our model of the coded V-BLAST is based, and its analysis can be found in \cite{Varanasi}\cite{Kost}\cite{Loyka2}.

The following optimization strategies are considered: optimum (per-stream) power allocation (OPA), optimum (per-stream) rate allocation (ORA), and optimum joint power/rate allocation (OPRA).
All of the above strategies can be instantaneous, i.e. a new per-stream allocation is found for each channel realization, or average, i.e. the allocation is found based on channel statistics and stays the same as long as the average SNR stays the same. This eliminates the effect of small-scale fading and tracks only large-scale variations, similarly to the uncoded system \cite{Varanasi}\cite{Kost}.

The system capacity\footnote{the system includes the channel and also its transmission/processing architecture, which forms an extended channel whose capacity is the system capacity.} (i.e. the sum of per-stream capacities) and outage probability are used as optimization criteria. For instantaneous optimization, the optimum power allocation $\boldsymbol{\alpha}$ is a function of both channel realization $\mathbf{H}$ and average SNR $\gamma_0$, $\boldsymbol{\alpha} = \boldsymbol{\alpha}(\mathbf{H}, \gamma_0)$. For average optimization, the optimum power allocation  $\boldsymbol{\overline{\alpha}}$ is a function of the average SNR $\gamma_0$ only, $\boldsymbol{\overline{\alpha}} = \boldsymbol{\overline{\alpha}}(\gamma_0)$ and stays constant as long as $\gamma_0$ is constant. In any case, the total power constraint applies, $\sum_{i = 1}^m\alpha_i = m$.

\section{Minimizing Outage Probability}
\label{sec:framework}
In this section, we consider a generic spatial multiplexing system (not only V-BLAST) operating in a fading channel of generic statistics (not only i.i.d. Rayleigh), which is quasi-static (non-ergodic or "slow fading"). The system outage probability, i.e. the probability that the system cannot support a target rate,  can be defined as follows \cite{Tse},
\begin{equation}
\label{eq2}
P_{out} = Pr \{C < R\}
\end{equation}
where $R$ is the target rate, and $C$ is the instantaneous (i.e. for given channel realization) system capacity. The following Lemma will be instrumental.
\begin{lemma}
\label{lemma1}
Consider two instantaneous optimization strategies $\boldsymbol{\alpha}^{1}$ and $\boldsymbol{\alpha}^{2}$ (either rate or power optimization can be used) such that $C^1=C(\boldsymbol{\alpha}^{1})\geq C(\boldsymbol{\alpha}^{2})=C^2$ $\forall \mathbf{H}$. Then the corresponding outage probabilities are related as
\begin{align}
P_{out}^1 = Pr \{C^1 < R\} \leq Pr \{C^2 < R\} = P_{out}^2
\end{align}
\end{lemma}
\begin{proof}
Define the outage sets $\mathcal{O}^{i} =  \{ \mathbf{H}: C^{i} < R\},\ i=1,2.$
The outage probabilities can then be rewritten as
$P_{out}^i = Pr \{ \mathbf{H} \in \mathcal{O}^{i}\}.$
From $C^1 \geq C^2$, it follows that $\mathcal{O}^1 \subseteq \mathcal{O}^2$, and thus $P_{out}^1 \leq P_{out}^2$. Q.E.D.
\end{proof}

Let us consider the following power (and/or rate) allocation strategies:
\begin{eqnarray}
\label{eq6}
\overline{ \boldsymbol{\alpha } }_C &=& \underset{\boldsymbol{\alpha}\left(\gamma_0\right) }{\operatorname{arg\,max}}{\overline{C}(\boldsymbol{\alpha})} \\
\label{eq7}
\overline{ \boldsymbol{\alpha } }_{out} &=& \underset{\boldsymbol{\alpha}\left(\gamma_0\right)}{\operatorname{arg\,min}}{P_{out}(\boldsymbol{\alpha})}\\
\label{eq8}
\boldsymbol{\alpha}_C &=& \underset{\boldsymbol{\alpha}\left(\gamma_0, \mathbf{H}\right)}{\operatorname{arg\,max}}
{C(\boldsymbol{\alpha})} \\
\label{eq9}
\boldsymbol{\alpha}_{out} &=& \underset{\boldsymbol{\alpha}\left(\gamma_0, \mathbf{H}\right)}{\operatorname{arg\,min}}{{P}_{out}(\boldsymbol{\alpha})}
\end{eqnarray}
where $\overline{C}$ is the mean (ergodic) capacity, and $C$, $\overline{C}$ and $P_{out}$ are considered as functions or functionals (in the case of instantaneous optimization) of the power allocation $\boldsymbol{\alpha}$; \eqref{eq6} and \eqref{eq7} correspond to the average optimization of the capacity and outage probability, and \eqref{eq8} and \eqref{eq9} correspond to the instantaneous optimizations, all subject to the total power constraint $\sum_{i = 1}^m\alpha_i = m$.
\begin{thm}
\label{thm1}
The outage probabilities of the optimization strategies in \eqref{eq6}-\eqref{eq9} are related as follows,
\begin{align}
\label{eq10}
Pr \{ C( \overline{\boldsymbol{\alpha}}_C ) < R \} &\geq Pr \{ C( \overline{\boldsymbol{\alpha} }_{out}  ) < R \}\\
\label{eq11}
                                                    &\geq Pr \{ C( \boldsymbol{\alpha}_{out}) < R \}=P_{out}^{opt}\\
\label{eq12}
                                                    &= Pr \{ C( \boldsymbol{\alpha}_C) < R \}
\end{align}
i.e. the instantaneous optimizations of the capacity and outage probability achieve the same lowest outage probability $P_{out}^{opt}$, the average optimization of the outage probability gives an intermediate result, and the average optimization of the capacity is worst in terms of the outage probability.
\end{thm}
\begin{proof}
 The inequality in \eqref{eq10} is by definition of $\overline{ \boldsymbol{\alpha} }_{out}$ (i.e. $\overline{ \boldsymbol{\alpha} }_{out}$ is the best average power allocation that minimizes $P_{out}$). The inequality in \eqref{eq11} is because the instantaneous optimization of $P_{out}$ cannot be worse than the average one. To prove the equality in \eqref{eq12} note that $P_{out}^{opt} \leq Pr \{ C( \boldsymbol{\alpha}_C) < R \}$ (by definition of $\boldsymbol{\alpha}_{out}$) and also that $C( \boldsymbol{\alpha}_C) \geq C( \boldsymbol{\alpha}_{out})$ (by definition of $\boldsymbol{\alpha}_C$). Using the last inequality and Lemma 1, $P_{out}^{opt} \geq Pr \{ C( \boldsymbol{\alpha}_C) < R \}$. Combining this with $P_{out}^{opt} \leq Pr \{ C( \boldsymbol{\alpha}_C) < R \}$ , \eqref{eq12} follows. It can be shown (by examples) that none of the inequalities in the Theorem can be strengthened to equalities.
\end{proof}

The high importance of \eqref{eq12} in Theorem 1 is due to the fact that while the problem in \eqref{eq9} is non-convex (as we show below, it has multiple solutions) and very difficult to deal with in general, either numerically or analytically, the problem in \eqref{eq8} is convex and has a well-known generic solution (via waterfilling) and, since the two are equivalent, this solution also applies to \eqref{eq9}

It follows from Theorem 1 that the outage capacities $C_{out}$, defined from $\Pr\{C<C_{out}\}=P_{out}$ (i.e. the maximum rate supported by the system with given outage probability $P_{out}$), of the optimization strategies above satisfy the inequalities
\begin{align}
\notag
C_{out}(\overline{\boldsymbol{\alpha}}_C )\leq C_{out}(\overline{\boldsymbol{\alpha}}_{out})
\leq C_{out}(\boldsymbol{\alpha}_C) = C_{out}(\boldsymbol{\alpha}_{out})
\end{align}

Let us now consider the problem in \eqref{eq9} in more detail.
\begin{thm}
\label{thm2}
Instantaneous optimization of the outage probability in \eqref{eq9} is a non-convex problem with an infinite number of solutions, one of which is the solution to the convex problem in \eqref{eq8}, i.e. via waterfilling.
\end{thm}
\begin{proof}
Let $C(\boldsymbol{\alpha})$ be the instantaneous capacity for given power allocation $\boldsymbol{\alpha}$ and $\mathcal{O}(\boldsymbol{\alpha})=\{\mathbf{H}: C(\boldsymbol{\alpha})< R\}$ be the corresponding outage set. To minimize $P_{out}=Pr\{\mathbf{H} \in \mathcal{O}(\boldsymbol{\alpha})\}$ is to minimize the set $\mathcal{O}(\boldsymbol{\alpha})$, which is obviously accomplished by maximizing $C(\boldsymbol{\alpha})$ for each $\mathbf{H}$, i.e. via \eqref{eq8}.
To demonstrate that this is not the only solution, we note that no optimization is necessary for all $\mathbf{H}$ such that $C(1,...,1)\geq R$ (i.e. if the unoptimized instantaneous capacity is not less than the target rate $R$) since such optimization, while increasing the capacity, will not shrink the outage set and, thus, will not reduce the outage probability. Thus, any power allocation can be used in such a case provided that the resulting capacity does not drop below $R$. The Corollary below gives some examples.
\end{proof}
\begin{cor}
\label{cor2}
Examples of several strategies to minimize $P_{out}$:
\begin{enumerate}
\item $\boldsymbol{\alpha} = \boldsymbol{\alpha}_{C}$ in \eqref{eq8} for any $\mathbf{H}$ optimizes both $C$ and $P_{out}$.
\item $\boldsymbol{\alpha} = \boldsymbol{\alpha}_{C}$ in \eqref{eq8} for $\mathbf{H} \in \mathcal{O}^{u}$, where $\mathcal{O}^{u} = \{\mathbf{H}: C(1,...,1)<R \}$ is the unoptimized outage set, and uniform otherwise optimizes $P_{out}$ but not necessarily $C$.
\item In an iterative numerical algorithm to find $\boldsymbol{\alpha}$, stop optimization as soon as $C(\boldsymbol{\alpha}) \geq R$.
\item $\boldsymbol{\alpha} = \boldsymbol{\alpha}_{C}$ for $\mathbf{H} \in \{\mathcal{O}^u - \mathcal{O}^{opt}\} $,
   where $\mathcal{O}^{opt} = \{\mathbf{H}: C(\boldsymbol{\alpha}_{C})<R \}$ is the optimized outage set (outage takes place even if the full optimum power allocation is performed), and uniform otherwise \footnote{Not particularly practical since one has to know beforehand the optimized outage set $\mathcal{O}^{opt}$.}.
\end{enumerate}
\end{cor}
Theorem \ref{thm2} implies that $P_{out}$ should not be relied on as the only performance/optimization criterion of a communication system. It should be used in conjunction with other performance measures, such as $C$. In general, a good optimization strategy should optimize both $P_{out}$ and $C$, as \#1 in Corollary \ref{cor2}.
\section{Optimum Allocation of Average Power/Rate}
\label{sec:OPA}
In this section, we consider the average optimum power allocation (OPA) among the streams with uniform rate allocation, and briefly discuss the effect of average rate allocation. The objective is to minimize the outage probability.

\subsection{Optimizing $P_{out}$ via the average OPA}
We assume the i.i.d. Rayleigh fading channel and fixed per-stream rate $R$, in which case the system outage takes place if one of the streams is in outage,
\begin{align}
\label{14}
P_{out} = 1 - \prod_{i = 1}^m ( 1 - Pr\{ C_i < R \} );
\end{align}
where $C_i = \ln(1+\alpha_i g_i\gamma_0)$ is the instantaneous capacity of $i$-th stream, $g_i = \vert \mathbf{h}_{i \perp } \vert^2$, $\mathbf{h}_{i \perp }$ is i-th column of the channel matrix projected onto the subspace perpendicular to the subspace of the yet-to-be-detected symbols, and the power allocation $\boldsymbol{\alpha}$ is the function of the average SNR $\gamma_0$ only. We also used the fact that in the i.i.d. Rayleigh fading channels different $g_i$ are independent of each other \cite{Varanasi}\cite{Loyka2}.
The outage probability of $i$-th stream is equal to the outage probability of $(n-m+i)$-th order MRC (follows from a slight modification of the results in \cite{Varanasi}\cite{Loyka2}):
\begin{align}
\label{15}
Pr\{ C_i < R \} &= Pr\{ \ln(1 + \alpha_i \gamma_i) < R \}
\\ \nonumber
& = F_{n-m+i} \left ( \frac{e^R - 1}{\alpha_i\gamma_0} \right )\\
& \approx \frac{(e^R - 1)^{n-m+i}}{(\alpha_i\gamma_0)^{n-m+i}(n-m+i)!}, \mbox{ } \frac{e^R - 1}{\alpha_i\gamma_0} \ll 1 \notag
\end{align}
where $F_k(x)=1-e^{-x}\sum_{l=0}^{l=k-1} x^l/l!$ is the outage probability of $k$-th order MRC.
In high SNR regime, first step dominates the outage probability,
\begin{align}
\label{16}
P_{out} \approx Pr \{ C_1 < R \}.
\end{align}
Similarly to the uncoded V-BLAST in \cite{Varanasi}\cite{Loyka2}, this behavior is due to the fact that the first stream in the coded system still has the lowest diversity order.

\begin{thm}
\label{thm3}
The optimum allocation of average power to minimize $P_{out}$ (i.e. the problem in \eqref{eq7}), at medium to high SNR, can be expressed as
\begin{align}
\label{17}
\overline{\alpha}_1 & \approx m-\sum_{i=2}^m { \overline{\alpha}_i } , \\
\overline{\alpha}_i & \approx \frac{ b_i } {\left( {4\gamma _0 } \right)^{\frac{i-1}{n-m+i+1}}}, \ i=2 \ldots m,\notag
\end{align}
where $b_i$ are numerical coefficients.
\end{thm}
\begin{proof}
\eqref{17} can be obtained using the Newton-Raphson method following the same approach as for the uncoded system in \cite{Kost}.
\end{proof}
\begin{cor}
\label{cor2}
The optimum power allocation in \eqref{17} behaves at high SNR as follows,
\begin{align}
\label{19}
\alpha_1 \approx m \gg \alpha_2 ... \gg \alpha_m,
\end{align}
i.e. most of the power goes to 1-st stream, with vanishingly small portions to higher-order streams (similarly to the case of uncoded V-BLAST in \cite{Kost}\cite{Varanasi}).
\end{cor}

The SNR gain of optimum power allocation $G$ is defined as the difference in the SNR
required to achieve the same error rate in the unoptimized system as in the optimized one \cite{Kost},

\begin{equation}
\label{def:Gain}
P_{out}\left( {\alpha _1^{opt} ,...,\alpha _m^{opt} } \right)=P_{out}\left( { G,...,G } \right)
\end{equation}
\begin{cor}
\label{cor:Gopa}
The SNR gain $G_{OPA}$ of the OPA (either instantaneous or average) is bounded as follows:
\begin{align*}
1 \leq G_{OPA} \leq m
\end{align*}
and monotonically increases with the average SNR.
\end{cor}
\begin{proof}
Follows from the definition of the SNR gain and Theorem \ref{thm3}, similar to the proof of Theorem 1 in \cite{Kost}.
\end{proof}

Diversity gain is defined as \cite{Tse}
\begin{align}
\label{21}
d = - \lim_{\gamma_0 \to \infty }\frac{\ln P_{out} }{\ln \gamma_0 };
\end{align}
\begin{cor}
The average OPA does not provide any additional diversity gain over the unoptimized system:
\begin{align*}
d_{OPA} = n-m+1 = d_u
\end{align*}
where $d_u$ is the diversity gain of the unoptimized system, which is the same as for the uncoded system \cite{Loyka2}\cite{Kost}.
\end{cor}
\begin{proof}
Follows immediately from Corollary 3.
\end{proof}
\subsection{Average Rate Allocation to minimize $P_{out}$}
Based on \eqref{15} and \eqref{16}, it is straightforward to show that the optimum allocation of the average rate does not increase the diversity order as long as all streams are active, which is required for the total rate to stay close to the system capacity. On the other hand, if the only goal is to minimize $P_{out}$ and the total rate is fixed, then the solution in \cite{Varanasi} applies (only slight modifications are required to account for the coding), but only one stream is then active at high SNR and the total rate is far below the capacity. Thus, we conclude that using only average power and/or rate allocation does not allow one to approach the capacity and to achieve high diversity simultaneously. We thus consider instantaneous optimization below.

\section{Optimum allocation of instantaneous rates for the coded V-BLAST}
\label{sec:ORA}
In this section, we study the optimum rate allocation (ORA) assuming the uniform power allocation, $\alpha_i = 1$.

Since capacity-achieving codes are used for each stream, the per-stream rate is set equal to the corresponding per-stream instantaneous capacity $C_i$ and the outage probability is given by
\begin{align}
\label{22}
P_{out}^{ORA} = Pr \left \{ \sum_i C_i < mR \right \},
\end{align}
i.e. the outage takes place only when the sum capacity is below the target rate $mR$ (compare to \eqref{14} where an outage takes place when any $C_i < R$). The diversity gain of such an allocation can be immediately characterized.
\begin{thm}
\label{thm4}
Diversity gain of the instantaneous ORA:
\begin{align}
\label{23}
d_{ORA} = m \left( n - \frac{m-1}{2} \right) > d_{OPA} = d_u.
\end{align}
i.e. there is an additional diversity compared to the average optimization.
\end{thm}
\begin{proof}
Denoting $C_{max} = \max_i{C_i}$, it is straightforward to bound $P_{out}^{ORA}$ as follows
\begin{align}
\label{25}
Pr \{C_{max} < R\} &\leq P_{out}^{ORA} \leq Pr \{C_{max} < mR\}.
\end{align}
Since per-stream SNRs $\gamma_i$ are independent, so are $C_i$, and the above inequality can be re-written as:
\begin{align}
\label{26}
\prod_{i} Pr\left\{C_i < R \right\} &\leq P_{out}^{ORA} \leq \prod_{i} Pr\left\{C_i < mR \right\}
\end{align}
Using \eqref{15}, \eqref{26} and taking the limit in \eqref{21}, one obtains
\begin{align}
\label{27}
 d_{ORA} = \sum_{i=1}^m(n-m+i) = m \left(n - \frac{m-1}{2} \right)
\end{align}
As diversity of the unoptimized system is $d_{u} = n-m+1$, the difference is $\Delta d_{ORA} =  d_{ORA} - d_{u} = \left(m-1 \right) \left( n- \frac{m}{2} + 1 \right ) > 0$, i.e. there is a significant advantage in using the instantaneous rate allocation based on the per-stream capacity (note that, since the per-stream rates are equal to the per-stream capacities, the sum rate equals to the system capacity and the system is optimal in this respect).
\end{proof}
\section{Joint Optimum Power/Rate Allocation (OPRA) }
\label{sec:OPRA}
In this section, we consider the joint optimum power and rate allocation.

\subsection{Optimizing $C$}
The conventional waterfilling algorithm cannot be applied to solve \eqref{eq8}, because, unlike receivers without SIC, ZF-VBLAST relies on successive interference cancelation and coefficients $\vert h_{i \perp} \vert$ change when any stream is turned off (no interference from the turned off stream needs to be nulled out, so the dimensionality of the sub-space to project out decreases). This results in \eqref{eq8} being non-convex problem in the case of V-BLAST. However, it can be split into $2^m-1$ convex sub-problems (one problem for each set of turned off streams), each of which can be solved by the application of the conventional WF algorithm, as described below.
\begin{thm}
\label{thm5}
The joint optimum allocation of instantaneous power/rate for the coded V-BLAST (i.e. \eqref{eq8}) is given by the Fractional Waterfilling Algorithm (FWF) below:
\begin{enumerate}
\item  Set  step $p = 1$, $C_{opt} = 0$, $p_{opt} = 1$, $k_{opt} = 1$ ($p-1$ is the number of turned off streams)
\item Set $k = 1$. The index $k$ determines what combination of $p-1$ streams is turned off, $k = 1,...,C_m^{p-1}$, where $C_m^{p-1}$ is the number of combinations of size $p-1$ out of $m$ original streams.
\item  Calculate the per-stream gains with interference from yet-to-be-detected symbols projected out, $\vert h_{i \perp}^{kp} \vert$ for $i =  1, \ldots, m - p + 1$ based on the channel matrix $ \mathbf H^{kp}$, containing k-th combination of $m-p+1$ columns of the original matrix $\mathbf H$.
\item Find the water level for $m-p+1$ active transmitters:
\begin{align}
\label{28}
\frac{1}{\lambda_{kp} } = \frac{1}{m-p+1}\left( \mathcal{P} + \frac{1}{\gamma_0} \sum_{i = 1}^{m-p+1}
\frac {1}{\vert h_{i \perp}^{kp} \vert ^2} \right),
\end{align}
 where $\lambda_{kp}$ is the Lagrange multiplier determining the "water level", $ \mathcal{P} = m$ is the total power.
\item Calculate the power allocation
\begin{align}
\label{29}
\alpha_i^{kp} = \left ( \frac{1} {\lambda_{kp} } - \frac {1} {\gamma_0 \vert h_{i \perp}^{kp} \vert^2 } \right )_+,\ i = 1... m-p+1,
\end{align}
where $(x)_+ = x$ if $x > 0$ and $0$ otherwise.
\item Calculate the per-stream and total capacities:
\begin{align*}
C^{kp}_i = \ln \left(1 + \alpha_i^{kp} \gamma_0 \vert h_{i \perp}^p \vert^2 \right),\
C^{kp} = \sum_{ \substack{i=1} }^{m-p+1} C^{kp}_i
\end{align*}
\item If $C^{kp} > C_{opt}$, set $C_{opt} = C^{kp}$, $p_{opt} = p$, $k_{opt} = k$.
\item Set $k = k + 1$ and go to Step 3 until $k = C_m^{p-1}$
\item Set $p = p + 1$ (this eliminates one more stream) and go to Step 2 until $p = m$.
\item $\alpha_i^{k_{opt}p_{opt}}$ and $C_i^{k_{opt}p_{opt}}$ are the optimum power and rate allocations.
\end{enumerate}
\end{thm}

While the FWF is more complex than the conventional WF, its incremental complexity is low for small $m$ (e.g. $m=2$). The following corollary shows that the FWF is close to the conventional one at high SNR.
%
\begin{cor}
\label{corWF1}
The fractional waterfilling algorithm converges to the conventional\footnote{"Conventional" in a sense that $h_{i \perp}$ are calculated only once and the waterfilling is done only once over the full set of streams. If some streams are assigned zero power, this does not lead to  recalculation of $h_{i \perp}$.} one at high SNR, when both produce the uniform power allocation for full-rank channels,  $\alpha^{opt}_i \to 1$ when $\gamma_0 \to \infty $.\footnote{for rank-deficient channels, both algorithms allocate no power to zero-gain dimensions and the same power to all active streams.}
\end{cor}
\begin{proof}
From \eqref{29}, $\alpha_i^{kp} \to \frac{1} {\lambda_{kp} } = \frac{m}{m-p+1}$ for all $k$, $p$ when $\gamma_0 \to \infty $. Let us denote $C^p = \max_k C^{kp}$, and $\boldsymbol h_{1 \perp}^p,..,\boldsymbol h_{m-p+1 \perp}^p$ is the subset of columns of $\textbf{H}$ corresponding to $C^p$. Since $\boldsymbol \alpha^p$ that yields the greatest $C^p$ is the solution, let us compare $C^{p}$ and $C^{p+1}$:
\begin{align*}
C^{p} - C^{p+1} \to \ln \left( \gamma_0 \frac
{\lambda_{p+1}^{m-p+1} \prod_{i=1}^{m-p+1}\vert h_{i \perp}^p \vert^2 }
{\lambda_{p}^{m-p} \prod_{i=1}^{m-p}\vert h_{i \perp}^{p+1} \vert^2} \right)
 > 0
\end{align*}
so $C^1 > C^2 > \ldots > C^m$, and fully dimensional system ($C^1$) is optimal when $\gamma_0 \to \infty$.
\end{proof}

The following corollaries follow immediately from the description of the FWF in Theorem 5.
\begin{cor}
\label{corWF2}
If the solution given by the FWF does not contain any $\alpha_i = 0$, it is the same as the solution given by the conventional WF.
\end{cor}
\begin{cor}
\label{corWF3}
If the solution given by the conventional WF contains any $\alpha_i = 0$, the solution given by the FWF also contains zeros, but not the opposite.
\end{cor}
\subsection{Optimizing $\overline{C}$}
In this section,  the mean (ergodic) capacity is used as the objective of
optimum allocation of the average power and rate, which applies to ergodic channels.
\begin{thm}
\label{thmAvC}
Assuming capacity-achieving codes for each stream with the rate equal to
the mean per-stream capacity, the optimum allocation of average power to
maximize the total average capacity at large SNR is:
\begin{align}
\alpha _i = \frac {1 - \frac{A_i}{\gamma_0}  } { 1 - \frac{1}{m\gamma_0} \sum_i A_i}, \ \gamma_0 \gg 1,
\end{align}
where $A_i$ are given by
\begin{align}
\label{eqAi}
A_i \approx
 \begin{cases}
 \ln \left( {\gamma _0 } \right) - \mathcal{E} , \ \mathcal{E} \approx 0.577, \ i =1-n+m \\
 \frac{1}{n-m+i -1}, \ i >1-n+m
 \end{cases}
\end{align}
\end{thm}
\begin{proof}
 by applying the method of Lagrange multipliers to resolve \eqref{eq6}, details are omitted due to the page limit.
\end{proof}
\begin{cor}
\label{corUniformGood}
For optimization of the mean capacity, $\alpha_i \approx 1$ at high SNR,
i.e. the uniform average power allocation is optimum at high SNR ($\gamma_0 \gg A_i$).
\end{cor}
\begin{proof}
Immediately follows from Theorem \ref{thmAvC}.
\end{proof}
Numerical analysis shows that the uniform power allocation is close to the optimum over a wide range of the average SNR except for very low values. It also follows that the instantaneous optimization cannot improve much the mean capacity (see also \cite{Tse}), unless the SNR is very low. This is due to the fact that the notion of mean (ergodic) capacity implies coding over long intervals of time so that the bad channel realization are compensated for by coding, and there is no much room for improvement left. For non-ergodic (slow fading) channels, the notions of outage probability and outage capacity apply and, in this case, the instantaneous optimization results in significant improvement since it improves the performance on bad channel realizations which dominate the outage probability, even at high SNR (the coding over time does not help in this case as the channel is quasi-static).

\subsection{Minimizing $P_{out}$ via the FWF}
Following Theorem 1, the FWF algorithm not only maximizes the instantaneous capacity $C$, but it also minimizes the outage probability $P_{out}$. In this section, we consider this improvement in $P_{out}$.
\begin{thm}
\label{thmSNRgainOPRA}
The SNR gain $G_{OPRA}$ of the optimum allocation of instantaneous power and rate (via the FWF) compared to the optimum rate allocaion only, defined from
\begin{align*}
P_{out}^{OPRA} (\gamma_0) = P_{out}^{ORA} (G_{OPRA} \gamma_0 ),
\end{align*}
is bounded as follows:
\begin{align*}
1 \leq G_{OPRA} \leq m
\end{align*}
\end{thm}
\begin{proof}
similar to that of Corollary 3.
\end{proof}
The following corollary is a straightforward consequence of this theorem.
\begin{cor}
\label{corOPRA}
Diversity gain of the OPRA via the FWF is the same as that of the ORA (see Theorem 4), i.e. the optimum power allocation on top of the rate allocation does not give any additional diversity gain, but only at most $m$-fold SNR gain.
\end{cor}

Thus, the instantaneous rate allocation is the most efficient of all techniques in terms of incremental improvement as it is the only technique that brings significant diversity gain and keeps the rate close to the capacity.

\section{Conclusion}
\label{sec:conclusion}
 A comparative analysis of the optimum power, rate and joint power-rate allocations, either instantaneous or average, for the coded V-BLAST have been presented. The instantaneous rate allocation is the most effective technique as it provides a significant diversity gain over the unoptimized V-BLAST. The power allocation can at most give an $m$-fold SNR gain and it does not provide any additional diversity (provided that all streams stay active, as required for near-capacity transmission). The power allocation on top of the rate allocation provides approximately the same incremental gain as the power allocation on top of the unoptimized (uniform, fixed rate) V-BLAST. Compact, closed-form expressions for the optimum allocation of average power have been derived to minimize the outage probability and it was shown that it has roughly the same effect on the outage probability for the coded V-BLAST as it for the uncoded one. Since the conventional waterfilling does not maximize the capacity of V-BLAST, the Fractional Waterfilling Algorithm (FWF), which optimizes simultaneously the instantaneous capacity and the outage probability of the coded V-BLAST, has been proposed, and its performance has been evaluated in terms of the diversity and SNR gains.

\bibliographystyle{IEEEbib}
\bibliography{my-bibliography-file}

\begin{thebibliography}{1}

\bibitem{FosLimits}G. J. Foschini and M. J. Gans, "On limits of wireless communications in a fading environment when using multiple antennas," \emph{Wireless Personal Commun.}, vol. 6, no. 3, pp. 311–335, Mar. 1998.
\bibitem{FosBLAST} G.J. Foschini et al, Simplified Processing for High Spectral Efficiency Wireless Communication Employing Multi-Element Arrays, \emph{IEEE Journal on Selected Areas in Communications}, v. 17, N. 11, pp. 1841-1852, Nov. 1999.
\bibitem{Kost} V. Kostina, S. Loyka, On Optimum Power Allocation for the V-BLAST,  \emph{IEEE Transactions on Communications},v.56, N.6, pp. 999-1012, June 2008
\bibitem{Varanasi} N. Prasad, M.K. Varanasi, Analysis of decision feedback detection for MIMO Rayleigh-fading channels and the optimization of power and rate allocations, \emph{IEEE Trans. Inform. Theory}, v. 50, No. 6, June 2004
\bibitem{Kalbasi} R. Kalbasi, D.D. Falconer, A.H. Banihashemi, Optimum power allocation for a V-BLAST system with two antennas at the transmitter, \emph{IEEE Communications Letters}, v. 9, No. 9, pp. 826-828, Sep. 2005
\bibitem{Nam} S.H. Nam et al, Transmit Power Allocation for a Modified V-BLAST System, \emph{IEEE Trans. Comm.}, v. 52, N. 7, pp.1074-1079, Jul. 2004.
\bibitem{Wang} N. Wang, S.D. Blostein, Minimum BER Transmit Power Allocation for MIMO Spatial Multiplexing Systems, \emph{2005 IEEE International Conference on Communications}, San Diego, California USA, v. 4, pp. 2282 -- 2286, May 2005.
\bibitem{Loyka} S. Loyka, F. Gagnon, Performance Analysis of the V-BLAST Algorithm: an Analytical Approach, \emph{IEEE Trans. Wireless Comm.}, v. 3, No. 4, pp. 1326-1337, July 2004.
\bibitem{Tse} D. Tse, P. Viswanath, Fundamentals of Wireless Communications, Cambridge University Press, 2005.
\bibitem{Loyka2} S. Loyka, F. Gagnon, V-BLAST without Optimal Ordering: Analytical Performance Evaluation for Rayleigh Fading Channels, IEEE Transactions on Communications, v. 54, N. 6, pp. 1109-1120, June 2006.

\end{thebibliography}

\end{document}